\documentclass[11pt]{article}  % Comment this line out if you need a4paper

\usepackage[margin=1in]{geometry}
\usepackage{amsmath} % assumes amsmath package installed
\usepackage{amssymb}  % assumes amsmath package installed
\usepackage{amsthm}
\usepackage{float}
\usepackage{graphicx}
\usepackage{gensymb}
\usepackage{comment}
\usepackage{subcaption}
\usepackage[hidelinks]{hyperref}
\newtheorem{theorem}{Theorem}
\newtheorem{assumptions}{Assumptions}
\title{\LARGE \bf
Decentralized Guidance and Control for Rendezvous and Docking with a Tumbling Target using Multiple Servicers
}

\author{
Jonathan Taylor Jönsson$^{1}$,
Sathyanarayanan Seshasayanan$^{1}$,\\
Sumeet Gajanan Satpute$^{1}$,
George~Nikolakopoulos$^{1}$\\
}

\date{}

\begin{document}

\maketitle

\begingroup
\renewcommand{\thefootnote}{}
\footnotetext{
$^{1}$Robotics and AI Team, Department of Computer Science,
Electrical and Space Engineering, Luleå University of Technology, Sweden.
{\tt\small Corresponding author's email: jonathanjonsson96@hotmail.se}
}
\addtocounter{footnote}{-1}
\endgroup

%%%%%%%%%%%%%%%%%%%%%%%%%%%%%%%%%%%%%%%%%%%%%%%%%%%%%%%%%%%%%%%%%%%%%%%%%%%%%%%%
\begin{abstract}
The growing trend of miniaturization in space, combined with an exponential increase in orbital debris, motivates a new operational paradigm, that is the deployment of multiple small servicers to collectively rendezvous, dock and service large tumbling targets. A single small spacecraft is often limited in thrust or reach for such targets, yet coordinating multiple autonomous spacecraft safely in close proximity remains an open challenge. This paper presents a proof-of-concept decentralized 6-DOF model predictive control (MPC) framework for multiple servicers to safely rendezvous and dock (RVD) to a tumbling target with no knowledge about each others predicted input over the prediction horizon. The proposed framework enforces soft docking conditions, inter-servicer collision avoidance, and obstacle avoidance through embedded constraints, while remaining computationally tractable via a linear time-invariant (LTI) formulation. Validation in MATLAB with three servicers and fixed obstacle arrangements confirms that all constraints are satisfied and successful docking is achieved.

%Due to the increasing number of satellites in Earth's orbit, there is a growing need for methods to remove tumbling space debris, on-orbit servicing, and recovery of malfunctioning satellites. This can be solved by using multiple autonomous servicer satellites that can rendezvous and dock (RVD) with precision. This paper presents a proof-of-concept decentralized 6-DOF model predictive control (MPC) framework for multiple servicers to RVD with a tumbling target. Constraints implemented into the model ensure soft-docking and collision avoidance. The model is validated in a MATLAB simulation environment with different fixed obstacle arrangements between the servicers and the target. The simulation results demonstrate that RVD successfully satisfies all the constraints imposed under limited information sharing.

\end{abstract}

%%%%%%%%%%%%%%%%%%%%%%%%%%%%%%%%%%%%%%%%%%%%%%%%%%%%%%%%%%%%%%%%%%%%%%%%%%%%%%%%
\section{Introduction}
The trend of miniaturization in space has led to an unprecedented surge in satellite deployments, with both governments and private companies launching increasingly capable small spacecraft at lower cost. At the same time, orbital debris has grown exponentially since the early days of spaceflight, with large defunct satellites and rocket bodies posing a particular problem, especially in Low Earth Orbits (LEO) \cite{goshu2025space}. Operations such as active debris removal, on-orbit servicing, and satellite recovery all share a common prerequisite, i.e.,  the ability to safely rendezvous and dock (RVD) with a tumbling uncooperative target. Combined with the growing availability of the small spacecraft further motivates the question of whether multiple such vehicles could be coordinated to collectively perform this task, distributing the efforts across a formation rather than relying on a single large spacecraft. Thus, the fundamental RVD capability forms one of the most important aspects to address. This paper addresses precisely that, the development and validation of the guidance and control framework needed to enable multiple small spacecraft to safely RVD with a tumbling target.       

The applications for un-manned servicers are not restricted to the issue of space debris. Servicers could also be used to refuel, upgrade, or repair current satellites in orbit \cite{flores2014review}. 

The RemoveDEBRIS mission was the first of its kind to demonstrate Active Debris Removal (ADR) capabilities in-orbit when it comes to rendezvous and capture by using nets, harpoons and drag sails in order to de-orbit debris \cite{aglietti2020removedebris}. Demonstrations of rendezvous and docking with a single spacecraft with cooperative targets include JAXA's satellite ETS-VII \cite{kawano2001result}, and DARPA'
s Orbital Express \cite{whelan2000darpa}. There has not yet been a flight demonstration in-orbit which can achieve rendezvous and docking using a single spacecraft with an uncooperative target due to the complexities of a time-varying docking surface in a tumbling state. 

Previous theoretical work that has been done involving multiple servicers is limited and tackles the problem using different control architectures and mostly nonlinear control methods. Zhang et al. \cite{zhang2025decentralized} describes the relative motion and a decentralized control pose strategy in the Lie group SE(3) framework, which allows for a global and coordinate-free approach as well as a practical approach to coupling the translational and rotational dynamics. Fourlas et al. \cite{fourlas2023vision} proposes a centralized control scheme using non-linear MPC (NMPC) in conjunction with a vision-based pose estimation framework. The same control architecture is adapted by Wang et al. \cite{wang2025manifold}, where they propose a manifold-based trajectory planning method that solves an optimal control problem under boundary values and uses rotated 3D-Dubins curves in order to satisfy convergence of the manifold-based method. In McCamish et al. \cite{mccamish2010autonomous}, a distributed control architecture is used with an LQR controller in conjunction with a artificial potential function (APF) for the robust collision avoidance capabilities for both fixed and moving obstacles. An attitude synchronization MPC framework is proposed by Stadler et al. \cite{stadler2026lightweight} which embeds a stabilizing feedback controller into the system dynamics that results in a linearly time-invariant (LTI) prediction model. 

Previous work surrounding RVD to a tumbling target using multiple servicers consists in large of nonlinear control approaches, leading to more computationally demanding methods. In order to reduce this demand, an LTI approach is investigated. This paper proposes a framework that is an extension of the dual-loop MPC work in \cite{stadler2026lightweight}. The main contributions are as follows:
\begin{itemize}
    \item A parameter-invariant and LTI 6-DOF MPC framework is proposed for the decentralized rendezvous and docking of multiple servicer spacecraft with a tumbling target without inter-servicer sharing of predicted inputs over the prediction horizon.
    \item A stabilizing control force variable for the translational part of the framework that consists of feed-forward and linear feedback terms.
    \item Finally, the framework is validated in a MATLAB simulation environment with three servicers and multiple fixed obstacle arrangements, demonstrating successful RVD while satisfying all imposed constraints. 
\end{itemize}
\section{Problem formulation and spacecraft dynamics}
This paper addresses the problem of autonomous RVD of multiple servicer spacecraft with a tumbling target spacecraft. The target is assumed to exhibit uncontrolled rotational motion, while a team of servicer spacecraft simultaneously approaches and docks at separate target docking points under a decentralized control architecture. 

Let the set of servicer spacecraft be defined as $S = \{1,...N_s\}$, where $N_s$ denotes the total number of servicers participating in the mission. The following assumptions are made throughout the formulation of the servicer spacecraft dynamics and the design of the control strategy:
\begin{assumptions} \label{assumption1}
Relative position and velocity information between servicers and the target are available through onboard navigation and estimation algorithms as presented in \cite{kim2007kalman, comellini2020vision}.
\end{assumptions}
\begin{assumptions} \label{assumptipon2}
The center of mass of the target spacecraft is assumed to follow a near-circular orbit.
\end{assumptions}
\begin{assumptions} \label{assumption3}
The relative distance between the servicer and the target is assumed to be much smaller than the distance from the Earth's center to the target.
\end{assumptions}
The reference frames used throughout this work are illustrated in Figure~\ref{fig:ref_system} and defined as follows: \begin{itemize} \item $\mathcal{F}_I=\{O_I,\hat{i}_I,\hat{j}_I,\hat{k}_I\}$: Earth-Centered Inertial (ECI) frame. \item $\mathcal{F}_O=\{O_O,\hat{i}_O,\hat{j}_O,\hat{k}_O\}$: Local-Vertical Local-Horizontal (LVLH) frame attached to the target spacecraft. \item $\mathcal{F}_{BT}=\{O_{BT},\hat{i}_{BT},\hat{j}_{BT},\hat{k}_{BT}\}$: Body-fixed frame of the target spacecraft. \item $\mathcal{F}_{BC}^i=\{O_{BC}^i,\hat{i}_{BC}^i,\hat{j}_{BC}^i,\hat{k}_{BC}^i\}$, $i\in S$: Body-fixed frame of the $i^{\mathrm{th}}$ servicer spacecraft. \end{itemize}
\begin{figure}[!hbtp]
    \centering
    \includegraphics[width=0.8\linewidth]{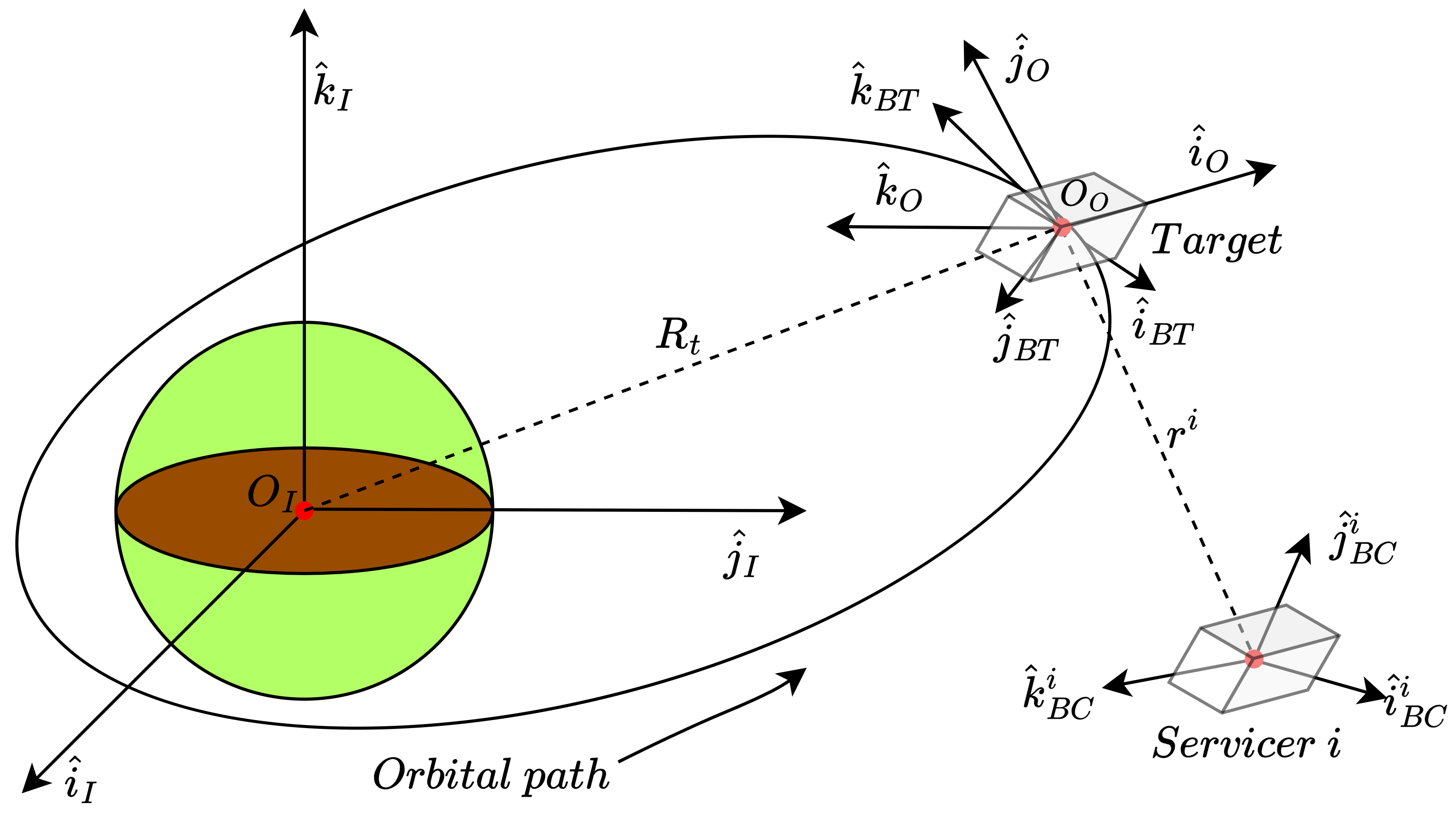}
    \caption{Illustration of the models reference frames}
    \label{fig:ref_system}
\end{figure}
\subsection{Relative motion dynamics}

The dynamics are derived under Assumptions~\ref{assumptipon2}--\ref{assumption3}. The relative translational motion of the $i^{\mathrm{th}}$ servicer is governed by the Clohessy-Wiltshire (CW) equations \cite{clohessy1960terminal}:

\begin{equation} \label{eq:linear}
\ddot{r}^i=A_v\dot{r}^i+A_r r^i+\frac{R_{OC}^if^i}{m},
\end{equation}

where

\begin{equation}
A_v=
\begin{bmatrix}
0&2n&0\\
-2n&0&0\\
0&0&0
\end{bmatrix},
\quad
A_r=
\begin{bmatrix}
3n^2&0&0\\
0&0&0\\
0&0&-n^2
\end{bmatrix}.
\end{equation}

Here, $r^i$ denotes the relative position of the $i^{\mathrm{th}}$ servicer in $\mathcal{F}_O$, $f^i= \begin{bmatrix} f^i_x,f^i_y ,f^i_z \end{bmatrix}^T$ is the control force vector expressed in $\mathcal{F}_{BC}^i$, $R_{OC}^i$ denotes the rotation matrix from $\mathcal{F}_{BC}^i \rightarrow \mathcal{F_{O}}$, $m$ is the servicer's mass, and $n=\sqrt{\mu/R_t^3}$ is the target orbital rate, where $\mu$ is the Earth's gravitational parameter.

\subsection{Attitude dynamics}
The rotational dynamics of the $i^{\mathrm{th}}$ servicer spacecraft are described by Euler's rigid-body equation
\begin{equation}  \label{eq:rotation}
\dot{\omega}^i = J^{-1}\left(\tau^i-\omega^i\times L\right), 
\end{equation}
where $\omega^i \in \mathbb{R}^3$ is the angular velocity vector expressed in $\mathcal{F}_{BC}^i$, $\tau^i=[\tau^i_x, \tau^i_y, \tau^i_z]^T$ is the control torque vector, $J$ is the inertia matrix of the servicer spacecraft, and $L$ is its angular momentum. It is assumed that reaction wheels or other momentum-exchanging devices keep the total angular momentum constant through the law of conservation. The translational and attitude dynamics presented above form the basis for the development of the distributed control strategy described in the following section.

The attitude of the $i^{\mathrm{th}}$ servicer evolves according to
\begin{equation}
    \dot{R}^i = R^i(w^i)^\times
\end{equation}
where $R^i$ denotes the attitude of the $i^{\mathrm{th}}$ servicer, and $(\cdot)^\times$ denotes the skew-symmetric operator.
\subsection{Control objective}
\label{control_objective}
The objective is to design a control law that enables each servicer spacecraft to rendezvous and dock to its assigned docking port on the tumbling target while avoiding collisions with the target and other servicers. Furthermore, the position, velocity, attitude, and angular velocity errors between each servicer and its assigned docking port are required to asymptotically converge to zero.

For the $i^{\mathrm{th}}$ servicer, the relative position and velocity errors are defined as

\begin{align}
\Delta r^i &= r^i-r_d^i,\\
\Delta \dot r^i &= \dot r^i-\dot r_d^i,
\end{align}
where $r_d^i=R_{OT}r_{\mathrm{dock}}^i$ denotes the desired docking location in the LVLH frame and $R_{OT}$ denotes the rotation matrix from $\mathcal{F}_{BT} \rightarrow \mathcal{F_{O}}$.

The attitude and angular velocity errors are defined as

\begin{align}
\Delta R^i &= R^iR_d^T,\\
\Delta \omega^i &= R_d(\omega^i-\omega_d),
\end{align}
where $R_d$ and $\omega_d$ denote the attitude and angular velocity of the target.

Figures~\ref{fig:rendezvous_fig}-\ref{fig:docking_fig} illustrates the RVD objective of multiple servicers to the tumbling target in a simplified 2D configuration for visualization purposes.

\begin{figure}[hbtp]
\centering
\begin{subfigure}{0.48\columnwidth}
    \centering
    \includegraphics[width=\linewidth]{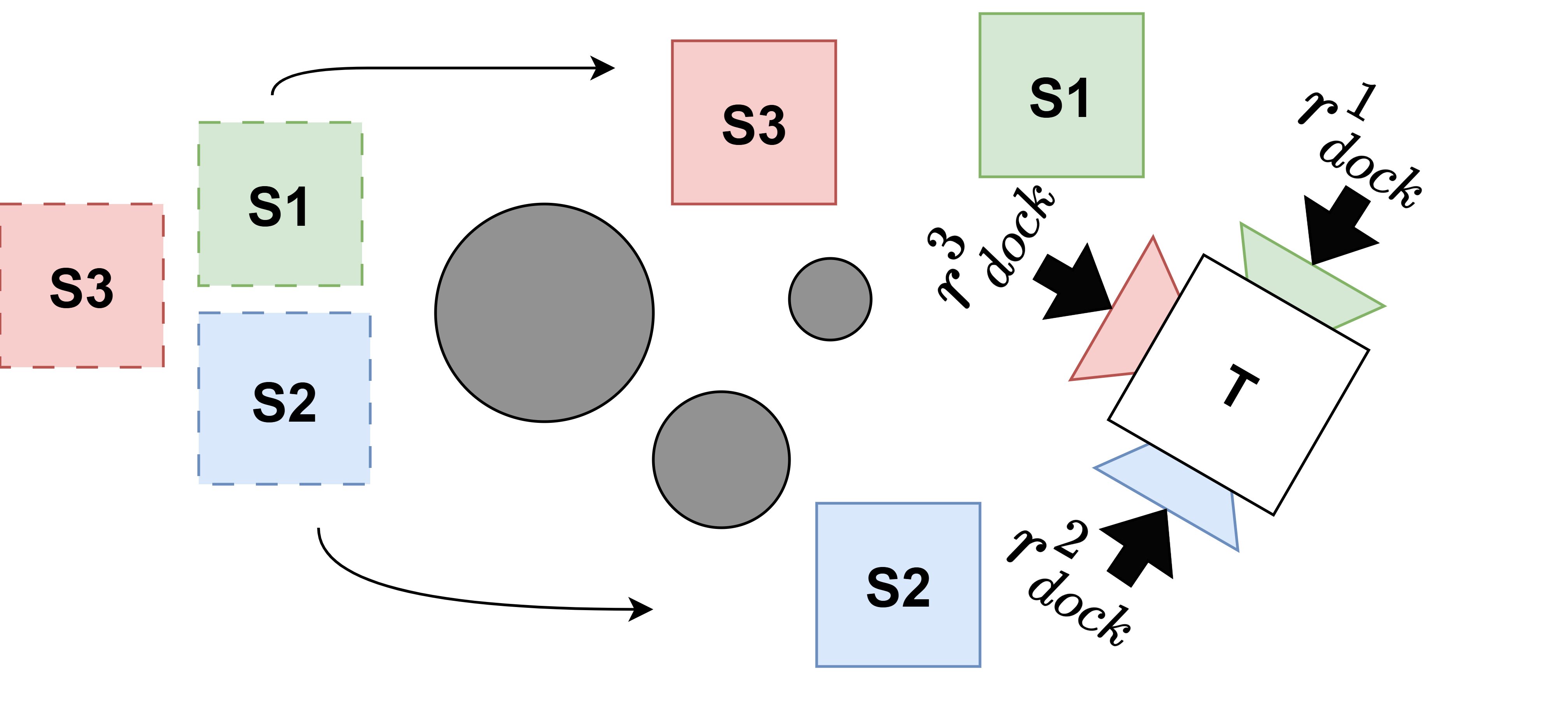}
    \caption{}
    \label{fig:rendezvous_fig}
\end{subfigure}
\hfill
\begin{subfigure}{0.48\columnwidth}
    \centering
    \includegraphics[width=\linewidth]{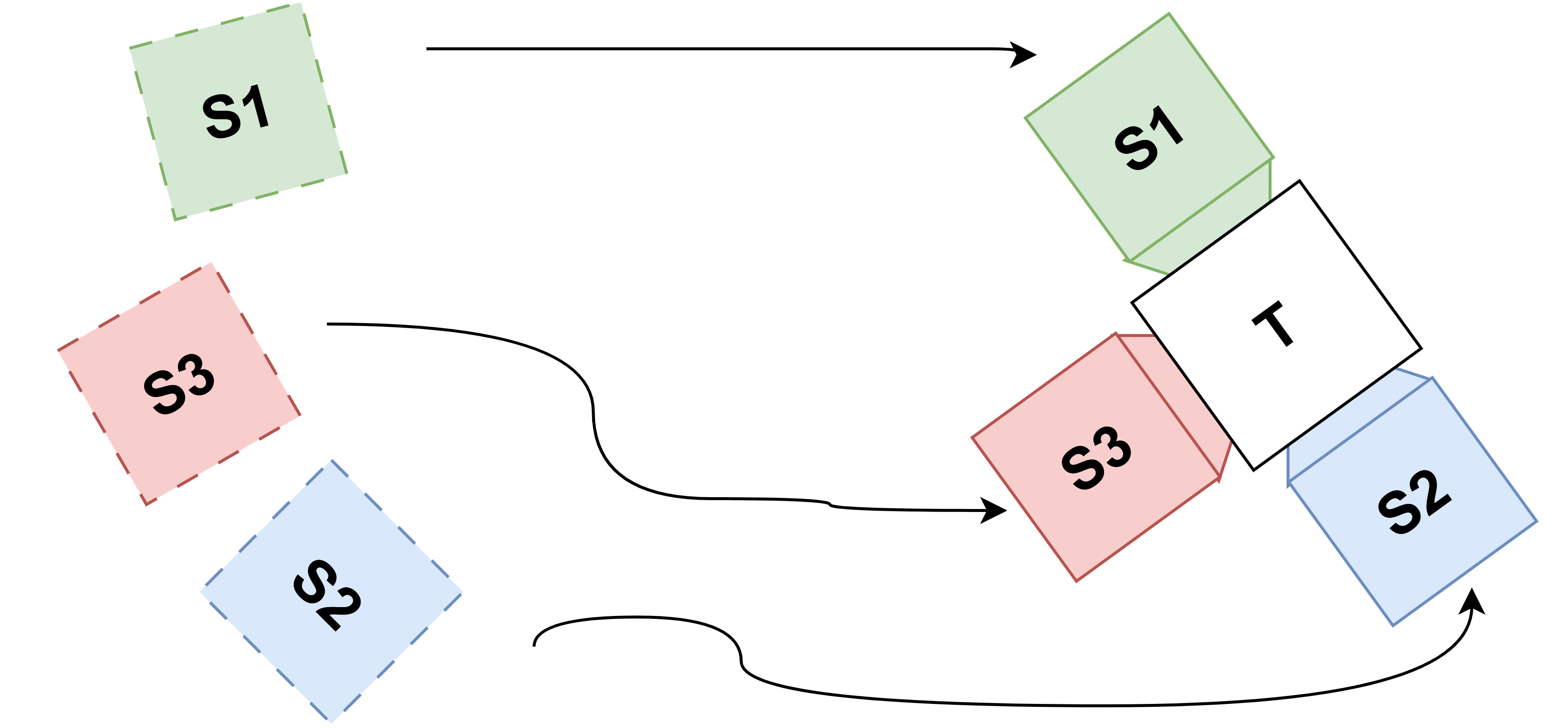}
    \caption{}
    \label{fig:docking_fig}
\end{subfigure}
\caption{(a) illustrates the rendezvous approach of the servicers denoted as S1, S2, and S3 towards the target spacecraft denoted as T with multiple obstacles in between, (b) illustrates the final docking phase in which each servicer's center of mass matches the position and velocity of its assigned target docking port, while also synchronizing its attitude and angular velocity with the target.}
\label{fig:RVD}
\end{figure}

\subsection{Constraint definition}
\label{constraint_section}
Constraints are implemented with the purpose of ensuring soft-docking, collision avoidance, and limiting the control inputs.
\subsubsection{Control input constraints}
The force and torque acting on the servicer is limited in order to constrain the amount of energy needed for the operation. $f^{max}$ and $\tau^{max}$ denote the maximum force and torque respectively that the actuators are constrained to.
\begin{equation}
    |f_{n}^i| \leq f^{max}, n = x,y,z
    \label{forceconstraint}
\end{equation}
\begin{equation}
    |\tau_{n}^i| \leq \tau^{max}, n = x,y,z
    \label{torqueconstraint}
\end{equation}
\subsubsection{Target collision avoidance constraint}
The dynamic target collision avoidance constraint from \cite{bashnick2023fast} uses a dynamic holding radius $r_{hold_{k}}^i$ at timestep $k$ which reduces once certain conditions are met.
\begin{equation}
    d^i_{t_k} = (r_{hold_{k}}^i)^2 - (r_k^i - r_{target_{k}})^T(r_k^i - r_{target_{k}}) \leq 0
    \label{nonlinear_CA}
\end{equation}
In order to linearize \eqref{nonlinear_CA}, an operating point is chosen on the boundary of the spherical "keep-out" zone that lies between the target's center and the position of the $i^{\mathrm{th}}$ servicer. The operating point is calculated by using the $i^{\mathrm{th}}$ servicers trajectory $r_{k-1}^i$ from the previous iteration. Since the targets position lies in the origin of frame $\mathcal{F_{O}}$, $r_{target_k}$ can be excluded from the final expression. Linearizing \eqref{nonlinear_CA} around the boundary point results in the following linearized expression.
\begin{equation}
    \tilde{d}^i_{t_k} = r_{hold_k}^i - \frac{(r^i_{k-1})^T}{||r^i_{k-1}||}r_k^i \leq 0
    \label{CA_finalconstraint}
\end{equation}

An initial holding radius $r_{hold,0}$ is defined for the start of the rendezvous and decreases based on if the servicers position is within a distance $\eta$ of the holding radius. If this condition is met, the holding radius at the next sampling instant is reduced by a factor of $\gamma$. The holding radius is reduced in this manner until the lower limit is met, i.e. docking distance to the target $r_{hold,min}$.

\subsubsection{Servicer collision avoidance constraint}
In order to avoid collision between the servicers a similar approach as \cite{li2017model} did for the target is defined. A servicer centered sphere with radius $r_{CAS}$ is used to define the "keep-out" zones of each servicer where $r_{s_k}^{i,j} = r_k^i -r_k^j = (x_{s_k}^{i,j},y_{s_k}^{i,j},z_{s_k}^{i,j})$ denotes the relative distance between the $i^{\mathrm{th}}$ servicer and a neighbouring $j^{\mathrm{th}}$ servicer in frame $\mathcal{F_{O}}$ at timestep k. 
\begin{equation}
    d_{s_k}^{i,j} = -((x_{s_k}^{i,j})^2 + (y_{s_k}^{i,j})^2 + (z_{s_k}^{i,j})^2) + r_{CAS}^2 \leq 0
    \label{CASquadratic}
\end{equation}

Since \eqref{CASquadratic} is quadratic, it is linearized around the relative position between two servicers from the previous iteration which results in the following linearized equation. 
\begin{equation}
    \tilde{d}_{s_k}^{i,j} = r_{CAS}^2 + (r_{s_{k|k-1}}^{i,j})^T(r_{s_{k|k-1}}^{i,j}) - 2(r_{s_{k|k-1}}^{i,j})^Tr_{s_{k|k}}^{i,j} \leq 0 
    \label{linear_CAS}
\end{equation}
Equation \eqref{linear_CAS} is imposed on the $i^{\mathrm{th}}$ servicer for every $j^{\mathrm{th}}$ servicer participating in the RVD.

\subsubsection{Angular velocity constraint}
The angular velocity of the chaser is constrained with respect to the maximum accepted angular velocity $\omega_i^{max}$.
\begin{equation}
    |\omega_{n}^i| \leq \omega^{max}, n =x,y,z
    \label{angvel_constraint}
\end{equation}
\subsubsection{Velocity constraint}
The velocity of the chaser is constrained in a similar way as the angular velocity in which it is constrained with respect to the maximum velocity accepted $v_{i}^{max}$ to ensure soft-docking.
\begin{equation}
    |\dot{r}_{n}^i| \leq v^{max}, n = x,y,z
    \label{vel_constraint}
\end{equation}

\section{Decentralized model predictive control design}
The decentralized MPC architecture for the $i^{\mathrm{th}}$ servicer can be seen in Figure~\ref{fig:control_architecture}. It consists of an inner loop with an error system that describes the closed-loop error dynamics in which the outer loop MPC operates on the resulting error state. The inner system becomes stable and closed-loop with the use of stabilizing feedback control variables resulting in a parameter- and time-invariant system. Linearizing the closed loop system results in a LTI prediction model. To achieve the control objective in Section~\ref{control_objective}, the control input is composed of two components: a MPC term and a stabilizing feedback term. The MPC component generates collision-free trajectories while satisfying input and safety constraints, whereas the stabilizing controller guarantees exponential convergence of the tracking errors to the desired docking state. The model is kept decentralized by sharing no information about the predicted inputs over the prediction horizon across the servicers. Instead, the servicers operate under Assumption~\ref{assumption1}, where the current relative position and velocity is available onboard at each timestep and propagated locally over the prediction horizon by each servicer. The model requires that each servicer has knowledge about the neighboring servicers respective docking coordinates. However, these are predetermined and remain constant throughout the RVD.
\begin{figure}[hbtp]
\centering
\includegraphics[width=\linewidth]{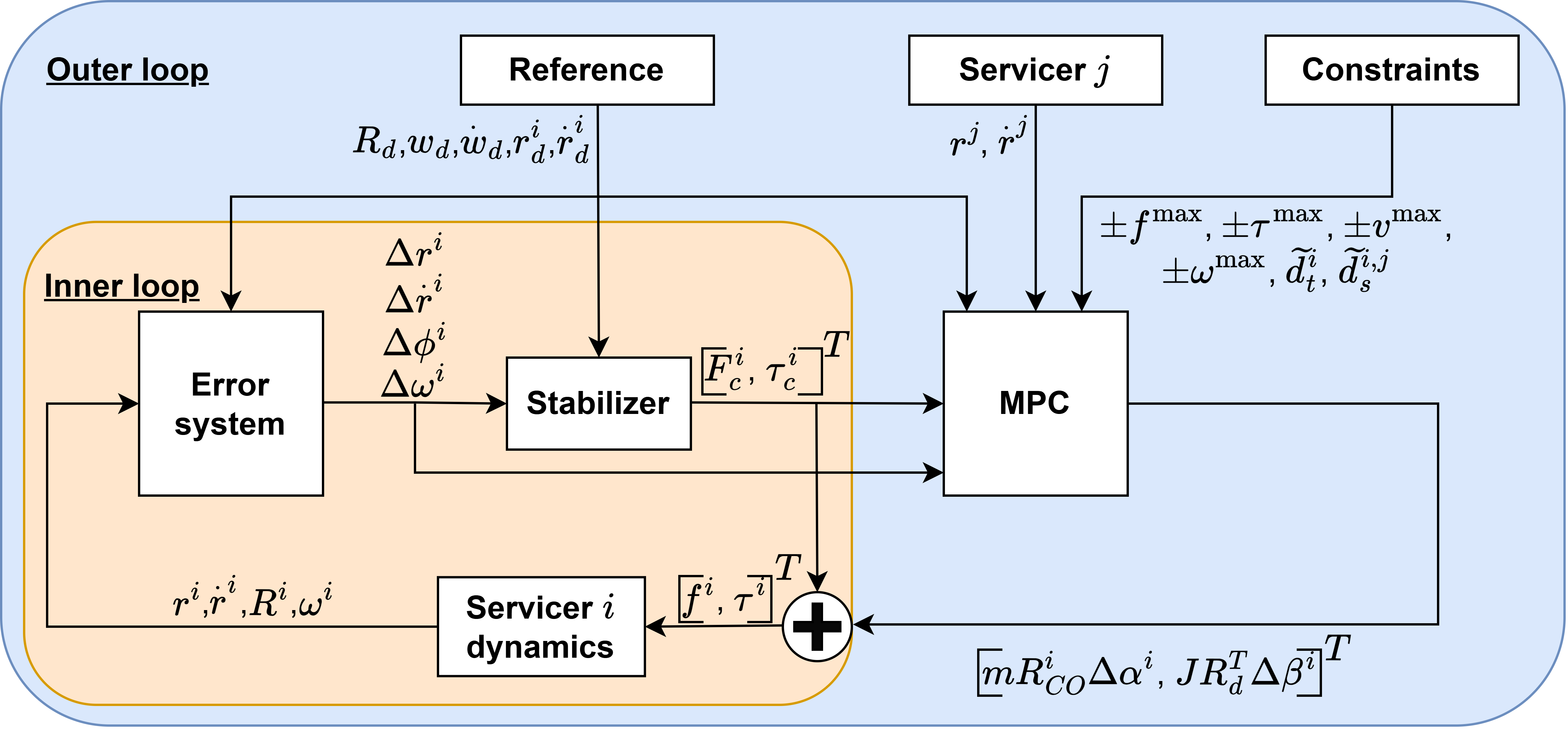}
\caption{Each $i^{\mathrm{th}}$ servicer functions independently using the same architecture, which consists of a dual-loop MPC scheme. With the help of the stabilizing control variables, the inner loop becomes a stable closed-loop system. Linearizing the closed-loop dynamics yields an LTI prediction model for the outer loop MPC. The relative position and velocity information of the $j^{\mathrm{th}}$ servicer are assumed available on-board.}
\label{fig:control_architecture}
\end{figure}

\subsection{Stabilizing feedback control}
\label{feedback_controllers}
In order to stabilize the inner system of the dual loop MPC, a stabilizing translational control force is chosen as,

\begin{equation}
F_c^i
=
m^iR_{CO}^i [\ddot r_d^i
-
A_rr_d^i
-
A_v\dot r_d^i -
(A_r+K_p)\Delta r^i
-
(A_v+K_v)\Delta\dot r^i]
\label{stabilizingforce}
\end{equation}
where $K_p = k_pI_3$ and $K_v = k_vI_3$.

\begin{theorem}
Consider the relative translational dynamics of the $i^{\mathrm{th}}$ servicer under the feedback control law \eqref{stabilizingforce}. If $K_p>0$ and $K_v>0$, then the equilibrium point

\[
(\Delta r^i,\Delta\dot r^i)=(0,0)
\]

is globally asymptotically stable.
\end{theorem}
\begin{proof}
Consider the Lyapunov candidate

\begin{equation}
V^i=
\frac{1}{2}(\Delta r^i)^TK_p\Delta r^i
+
\frac{1}{2}(\Delta\dot r^i)^T\Delta\dot r^i,
\end{equation}

which is positive definite for $K_p>0$. Substituting the control law \eqref{stabilizingforce} into the translational error dynamics yields

\begin{equation}
\Delta\ddot r^i=-K_p\Delta r^i-K_v\Delta\dot r^i.
\end{equation}

The time derivative of $V^i$ along the closed-loop trajectories is

\begin{equation}
\dot V^i
=
-(\Delta\dot r^i)^TK_v\Delta\dot r^i
\leq 0.
\end{equation}

Since $K_v>0$, $\dot V^i=0$ only when $\Delta\dot r^i=0$. From the closed-loop dynamics, this implies $K_p\Delta r^i=0$, and therefore $\Delta r^i=0$. By LaSalle's invariance principle \cite{lasalle1960some}, the equilibrium $(\Delta r^i,\Delta\dot r^i)=(0,0)$ is globally asymptotically stable.
\end{proof}

Similarly for the $i^{\mathrm{th}}$ servicer, the stabilizing control torque is derived in \cite{stadler2026lightweight} as,

\begin{equation}
\tau_c^i
=
\omega^i\times L
+
J(\dot{\omega}_d-\omega_d\times\omega^i)
-
JR_d^T
\left(
K_c\Delta\phi^i
+
K_w\Delta\omega^i
\right),
\label{stabilizingtorque}
\end{equation}
where $\Delta\phi^i := \log(\Delta R^i)$ is the error rotation vector and $K_c = k_cI_3$, $K_w = k_wI_3$.

\subsection{Error dynamics}
The error dynamics of the $i^{\mathrm{th}}$ servicer that are used in the prediction model, can be described by $\Delta \dot{r}_i, \Delta \ddot{r}_i, \Delta \dot{\phi}, \Delta \dot{\omega}$, denoting the time derivatives of the relative position and velocity error, as well as the attitude and angular velocity error between the $i^{\mathrm{th}}$ servicer and target respectively. The error dynamics are obtained through the stabilizing control variables in Section~\ref{feedback_controllers}. 

\begin{align}
\Delta\ddot{r}^i &= -K_p\Delta r^i -K_v\Delta \dot{r}^i + I_3\Delta\alpha^i ,
& 
\Delta\dot{r}^i &= I_3\Delta\dot{r}^i,
\label{ddot_deltar}
\end{align}
\begin{align}
\Delta\dot{\omega}^i &= -K_c\Delta \phi^i -K_w\Delta \omega^i + I_3\Delta\beta^i,
&
\Delta\dot{\phi}^i &= I_3\Delta\omega^i,
\label{dot_deltaw}
\end{align}

where $\Delta\alpha^i$ and $\Delta\beta^i$ are the translational and rotational correction inputs generated by the MPC,

For the MPC formulation, \eqref{ddot_deltar}-\eqref{dot_deltaw} are rewritten as,

\begin{equation}
\dot{X}^i = AX^i + BU^i,
\label{xdot_ss}
\end{equation}

with

\begin{align}
X^i &= \begin{bmatrix}
\Delta r^i &
\Delta\dot r^i &
\Delta\phi^i &
\Delta\omega^i
\end{bmatrix}^{T},
&
U^i &= \begin{bmatrix}
\Delta\alpha^i &
\Delta\beta^i
\end{bmatrix}^{T}.
\end{align}

The system matrices are given by

\begin{equation}
A=
\begin{bmatrix}
A_p & 0\\
0 & A_a
\end{bmatrix},
\qquad
B=
\begin{bmatrix}
B_p & 0\\
0 & B_a
\end{bmatrix},
\label{system_matrices}
\end{equation}

where

\begin{align}
A_p &=
\begin{bmatrix}
0 & I_3\\
-K_p & -K_v
\end{bmatrix},
&
A_a &=
\begin{bmatrix}
0 & I_3\\
-K_c & -K_w
\end{bmatrix},
&
B_p &= B_a =
\begin{bmatrix}
0\\
I_3
\end{bmatrix}.
\end{align}

In order to implement this using MPC, \eqref{xdot_ss} is discretized with a sampling period $T_s$ in which a time-discrete model can be defined as,

\begin{equation}
    X_{k+1}^i = A_dX_k^i + B_dU_k^i
\end{equation}

where $X_k^i$ and $U_k^i$ denote the state and control vectors at timestep $k$ respectively. The continuous matrices $A$ and $B$ in \eqref{system_matrices} are discretized with a zero-order hold (ZOH) resulting in the discrete-time matrices $A_d$ and $B_d$.

\subsection{MPC formulation}
The MPC has a prediction horizon of $N_p$ with a timestep of $\Delta t$. It consists of a state cost matrix $Q$ and an input cost matrix $R$ which are both symmetric and positive semi-definite. The terminal cost $P$ is calculated by solving the discrete algebraic Riccati equation. Since the MPC uses the error states between target and the $i^{\mathrm{th}}$ servicer, the constraints are mapped to constrain the physical parameters of interest seen in Section~\ref{constraint_section} for each of the $i^{\mathrm{th}}$ servicer.
\begin{equation}
    \underset{U_k^i}{min}\ ||X^i_{Np}||^2_{P} + \sum_{k <N_p}||X^i_k||^2_{Q} + ||U^i_k||^2_{R},\ i\in S
\end{equation}
\begin{equation}
    s.t. \ \ \ \ \ \ \  X^i_{k|k} = X^i_k
\end{equation}
\begin{equation}
    \ \ \ \ \ \ \ \ \ \ X^i_{k+1}= A_dX^i_k + B_dU^i_k
\end{equation}
\begin{align}
    f^i&\in [-f^{max},f^{max}], & \tau^i&\in [-\tau^{max},\tau^{max}],
\end{align}
\begin{align}
    \omega^i&\in [-\omega^{max},\omega^{max}], & \dot{r}^i&\in [-v^{max},v^{max}],
\end{align}
\begin{align}
    \tilde{d}^i_{t,k}&\leq 0, & \tilde{d}_{s_k}^{i,j}\leq 0, j \in S, j\neq i,
\end{align}

\section{Simulation results}
\subsection{Simulation framework}
Simulation parameters are chosen to present a proof-of-concept RVD scenario. The MPC weighting matrices and controller gains correlating to the translational part was chosen empirically to provide stable tracking. The matrices and gains associated with the attitude part were based on the findings in \cite{stadler2026lightweight}. The input constraints, and the servicer collision avoidance constraints were chosen for the proof-of-concept scenario. The target collision avoidance constraint is based off of the target dimensions in order to accurately simulate RVD to the target. The RVD procedure is simulated using fixed obstacles which are meant to be in the servicers pathway to the target. The "keep-out zone" radius for obstacle 1-3 (obstacle 1 is closest to the servicers, obstacle 3 is closest to the target) is $1.6\ m,\ 1.2\ m,\ 0.8\ m$ respectively. The servicer's angular velocity is initialized as $w_0^i = [0,0,0]^T\ rad/s$ for $i = 1,2,3$. The positive feedback gains are set to $k_p = 0.55$, $k_v = \sqrt{4k_p}$, $k_c = 0.36$, and $k_w = \sqrt{4k_c}$. The input matrices are defined as $R_p = R_a = 100I_3$.
\begin{align}
    R &= \begin{bmatrix}
           R_p & 0_{3\times 3} \\
           0_{3\times 3} & R_a \\
         \end{bmatrix}
\end{align}
The weighting matrices are defined in the following equations.
\begin{align}
    Q_{a} &= Q_p = \begin{bmatrix}
           10I_3 & 0_{3\times 3} \\
           0_{3\times 3} & I_3 \\
         \end{bmatrix}
    &
    Q &= \begin{bmatrix}
           Q_p & 0_{6\times 6} \\
           0_{6\times 6} & Q_a \\
         \end{bmatrix}
\end{align}
The simulation is initialized by the parameters found in Table~\ref{tab:sim_parameters_general} and Table~\ref{tab:sim_parameters_target_servicers}.
\begin{table}[hbtp]
\centering
\begin{tabular}{|c|c |}
\hline
\bf{Parameter} & \bf{Value} \\
\hline
Simulation time & 40 $s$ \\
Prediction horizon, $N_p$ & 11\\
Timestep, $\Delta t$ & 0.1 $s$ \\
Number of servicers, $N_s$ & 3  \\
Servicer's safety distance, $r_{CAS}$ & 1 $m$ \\
Max force, $f^{max}$ & 120 $N$ \\
Max torque, $\tau^{max}$ & 40 $Nm$ \\
Max velocity, $v^{max}$ & 1 $m/s$ \\
Max angular velocity, $\omega^{max}$ & 0.5 $rad/s$ \\
Initial holding radius, $r_{hold,0}$ & 3 $m$ \\
Min holding radius, $r_{hold,min}$ & 1 $m$ \\
Scaling factor, $\gamma$ & 0.95 \\
Evaluating distance, $\eta$ & $\sqrt{0.5}$ $m$ \\
Servicer's mass, $m$ & 150 $kg$ \\
Servicer's inertia matrix, $J$ & $J=diag(85,\ 94,\ 92)$ $kg\cdot m^2$ \\
Servicer's angular momentum, $L$ &  $(34, 28.2, 0)^T$ $kg\cdot m^2/s$ \\
\hline
\end{tabular}
\caption{General simulation parameters}
\label{tab:sim_parameters_general}
\end{table}
\begin{table}[hbtp]
\centering
\begin{tabular}{|c|c| }
\hline
\bf{Parameter} & \bf{Value}  \\
\hline
\multicolumn{2}{|c|}{\textbf{Target}} \\
\hline
Semi-major axis & 7178160 $m$ \\
Eccentricity & 0.002  \\
Inclination & 60$^\circ$ \\
RAAN & 30$^\circ$ \\
Argument of pericentre & 0$^\circ$ \\
True anomaly & 0$^\circ$\\
Angular velocity, $\omega_d$ & $[0.4,\ 0.3,\ 0.2]^T$ $rad/s$ \\
Angular acceleration, $\dot{\omega}_d$ & $[0,\ 0,\ 0]^T$ $rad/s^2$ \\
Initial rotation vector, $\phi_{d,0}$ & $[2.3,\ -2,\ 0]^T$ $rad$ \\
\hline
\multicolumn{2}{|c|}{$\mathbf{1^{\mathrm{st}}}$ \textbf{Servicer}} \\
\hline
Initial rotation vector, $\phi_0^1$ & $[0,\ 0,\ 0]^T$ $rad$ \\
Initial relative position, $r_0^1$ & $[-8,\ 15,\ 16]^T$ $m$\\
Initial relative velocity, $\dot{r}_0^1$ & $[0,\ 0.01,\ 0]^T$ $m/s$ \\
Target docking port, $r_{dock}^1$ & $[0,\ 0,\ 1]^T$ $m$ \\
\hline
\multicolumn{2}{|c|}{$\mathbf{2^{\mathrm{nd}}}$ \textbf{Servicer}} \\
\hline
Initial rotation vector, $\phi_0^2$ & $[0.4,\ 0.6,\ -0.3]^T$  $rad$ \\
Initial relative position, $r_0^2$ &$[-8,\ 17,\ 16]^T$ $m$\\
Initial relative velocity, $\dot{r}_0^2$ & $[0.01,\ 0,\ 0]^T$ $m/s$ \\
Target docking port, $r_{dock}^2$ & $[0,\ 0,\ -1]^T$ $m$ \\
\hline
\multicolumn{2}{|c|}{$\mathbf{3^{\mathrm{rd}}}$ \textbf{Servicer}} \\
\hline
Initial rotation vector, $\phi_0^3$ & $[1,\ -0.8,\ 0.2]^T$ $rad$ \\
Initial relative position, $r_0^3$ & $[-10,\ 15,\ 16]^T$ $m$\\
Initial relative velocity, $\dot{r}_0^3$ & $[0,\ 0,\ 0.01]^T$ $m/s$ \\
Target docking port, $r_{dock}^3$ & $[1,\ 0,\ 0]^T$ $m$ \\
\hline
\end{tabular}
\caption{Target and servicer simulation parameters}
\label{tab:sim_parameters_target_servicers}
\end{table}

\begin{figure}[!t]
\centering

\begin{subfigure}{0.235\textwidth}
    \centering
    \includegraphics[width=\linewidth]{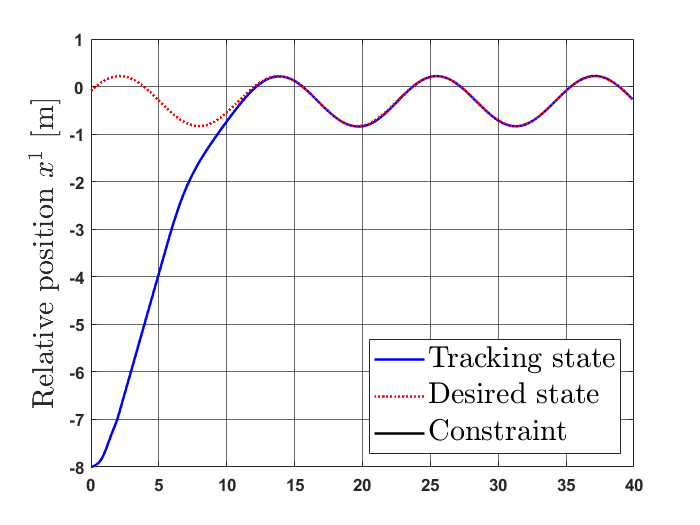}
    \label{fig:xpos_s1}
\end{subfigure}
\hfill
\begin{subfigure}{0.235\textwidth}
    \centering
    \includegraphics[width=\linewidth]{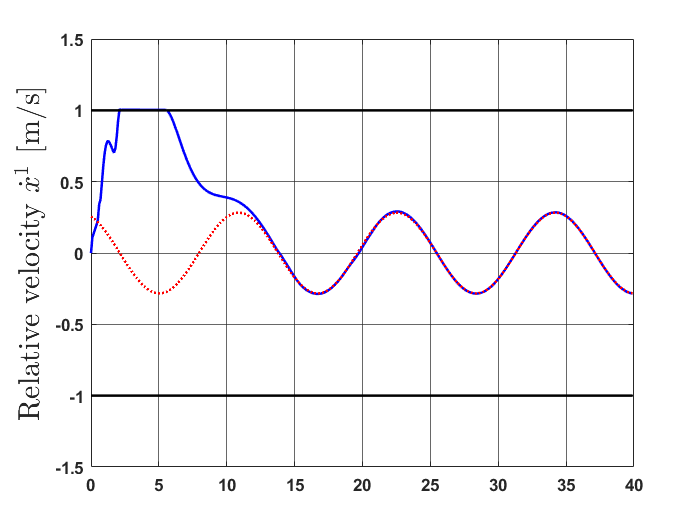}
    \label{fig:xvel_s1}
\end{subfigure}
\hfill
\begin{subfigure}{0.235\textwidth}
    \centering
    \includegraphics[width=\linewidth]{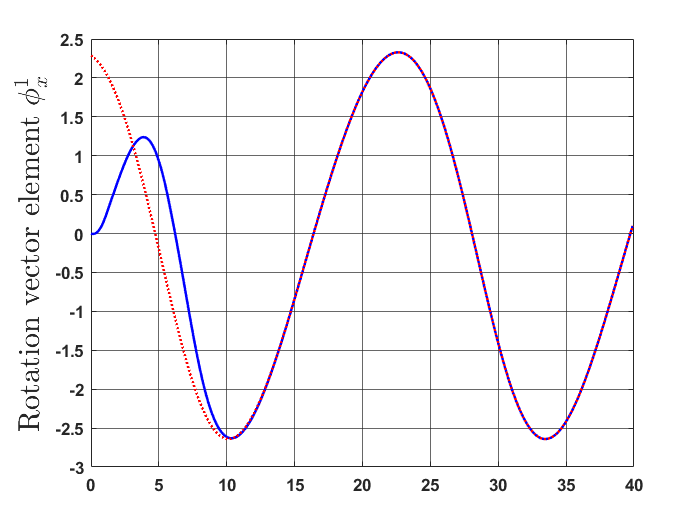}
    \label{fig:phix_s1}
\end{subfigure}
\hfill
\begin{subfigure}{0.235\textwidth}
    \centering
    \includegraphics[width=\linewidth]{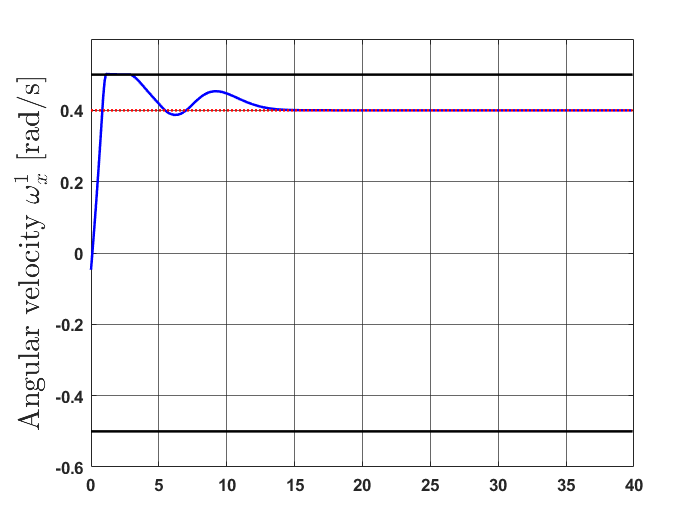}
    \label{fig:avelx_s1}
\end{subfigure}

\vspace{-0.60cm}

\begin{subfigure}{0.235\textwidth}
    \centering
    \includegraphics[width=\linewidth]{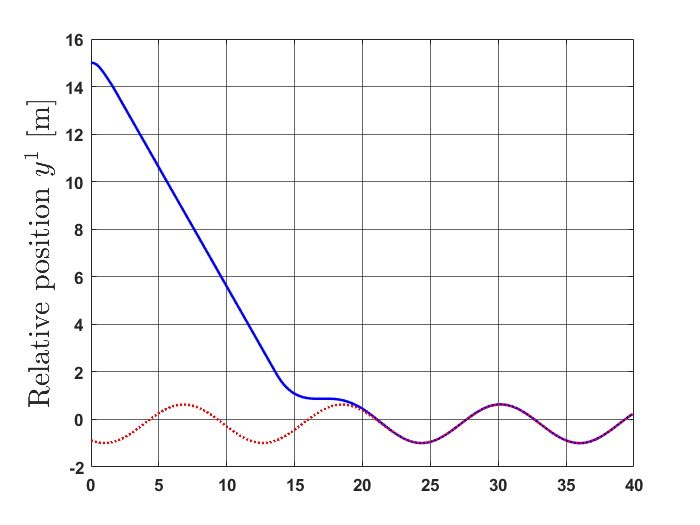}
    \label{fig:ypos_s1}
\end{subfigure}
\hfill
\begin{subfigure}{0.235\textwidth}
    \centering
    \includegraphics[width=\linewidth]{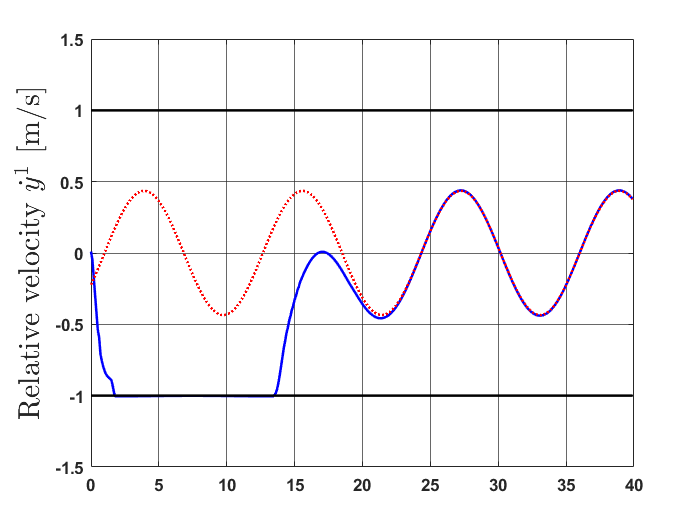}
    \label{fig:yvel_s1}
\end{subfigure}
\hfill
\begin{subfigure}{0.235\textwidth}
    \centering
    \includegraphics[width=\linewidth]{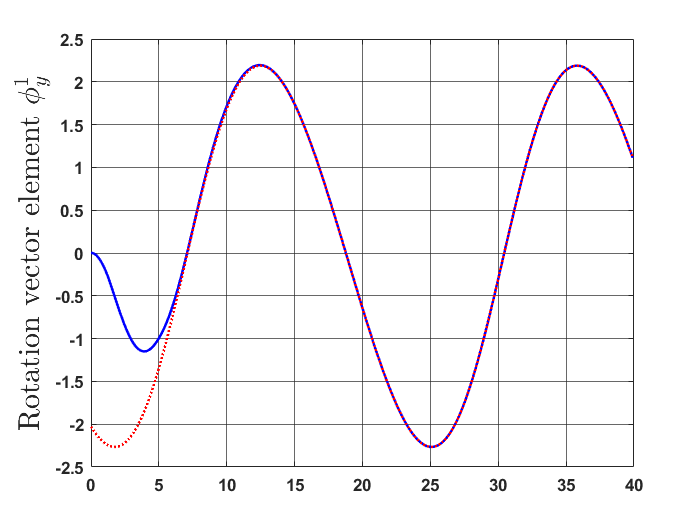}
    \label{fig:phiy_s1}
\end{subfigure}
\hfill
\begin{subfigure}{0.235\textwidth}
    \centering
    \includegraphics[width=\linewidth]{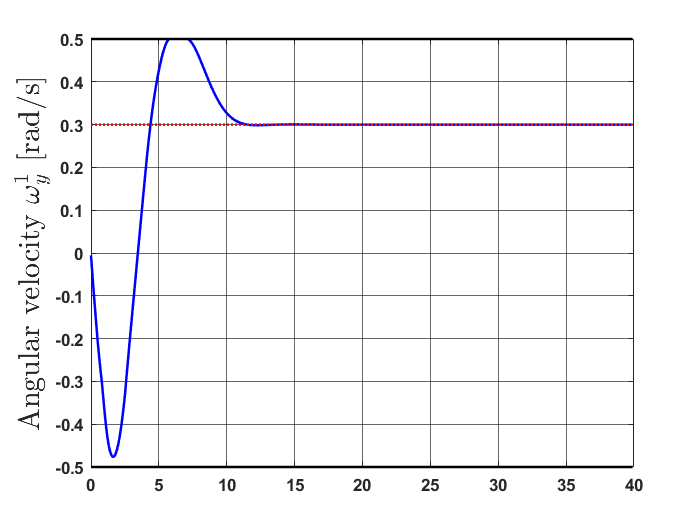}
    \label{fig:avely_s1}
\end{subfigure}

\vspace{-0.60cm}

\begin{subfigure}{0.235\textwidth}
    \centering
    \includegraphics[width=\linewidth]{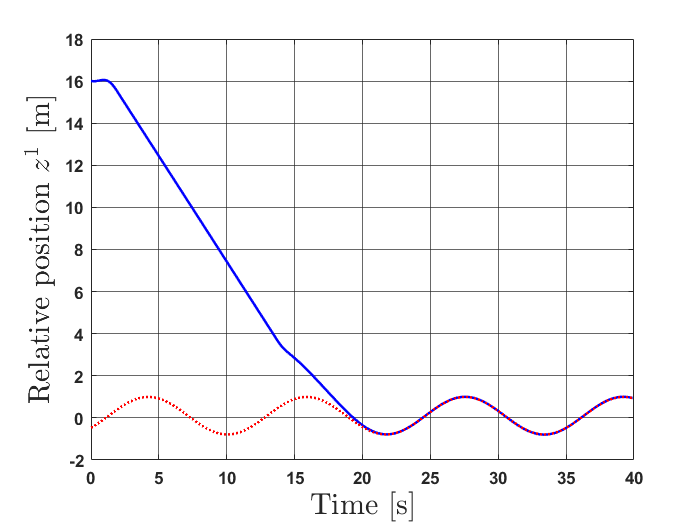}
    \caption{}
    \label{fig:pos_s1}
\end{subfigure}
\hfill
\begin{subfigure}{0.235\textwidth}
    \centering
    \includegraphics[width=\linewidth]{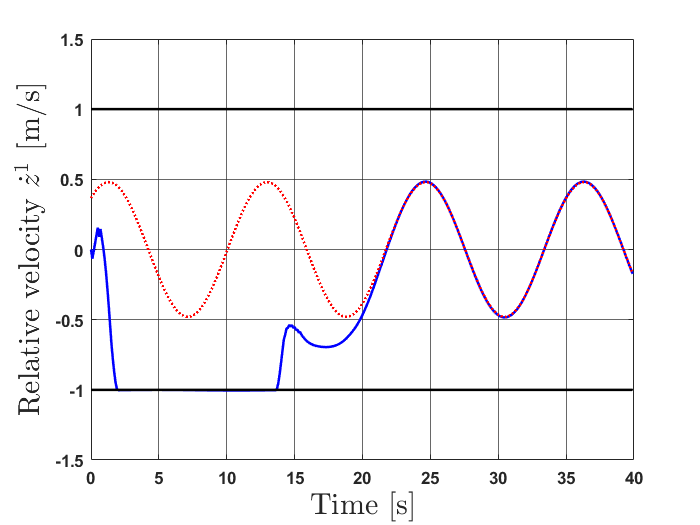}
    \caption{}
    \label{fig:vel_s1}
\end{subfigure}
\hfill
\begin{subfigure}{0.235\textwidth}
    \centering
    \includegraphics[width=\linewidth]{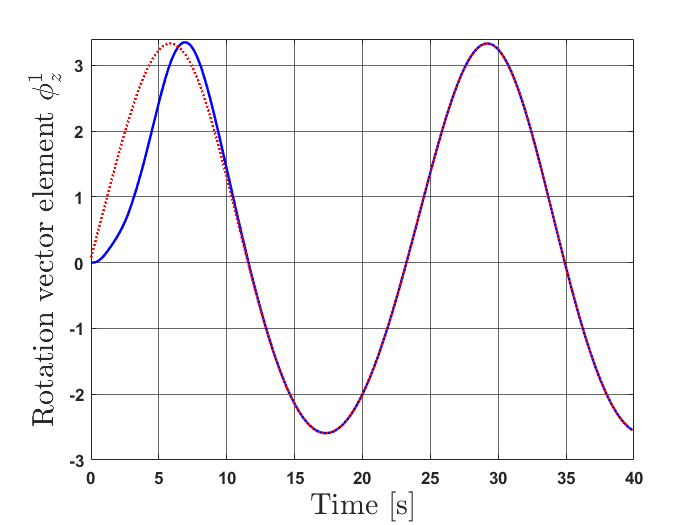}
    \caption{}
    \label{fig:att_s1}
\end{subfigure}
\hfill
\begin{subfigure}{0.235\textwidth}
    \centering
    \includegraphics[width=\linewidth]{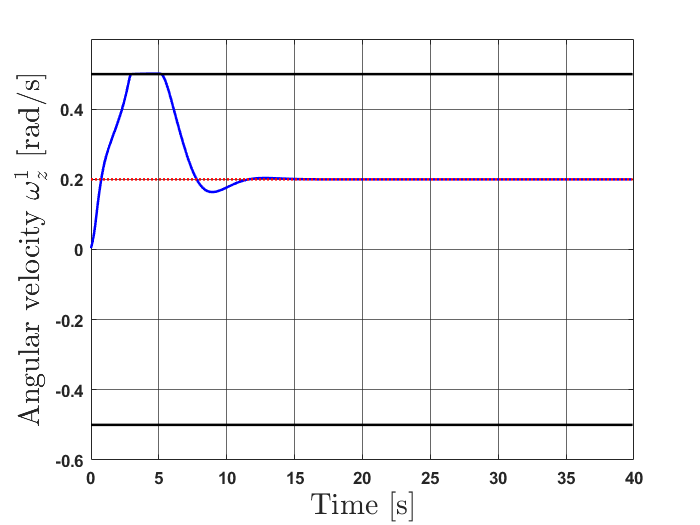}
    \caption{}
    \label{fig:avel_s1}
\end{subfigure}

\caption{Time history of servicer 1's tracking states throughout the RVD. (a)--(d) show the servicer tracking the desired states of relative position, relative velocity, attitude, and angular velocity respectively. All states are successfully tracked, and in (b) and (d), the velocity constraints are satisfied throughout the RVD.}
\label{fig:s1_states}
\end{figure}
\begin{figure}[!t]
\centering

\begin{subfigure}{0.32\textwidth}
    \centering
    \includegraphics[width=\linewidth]{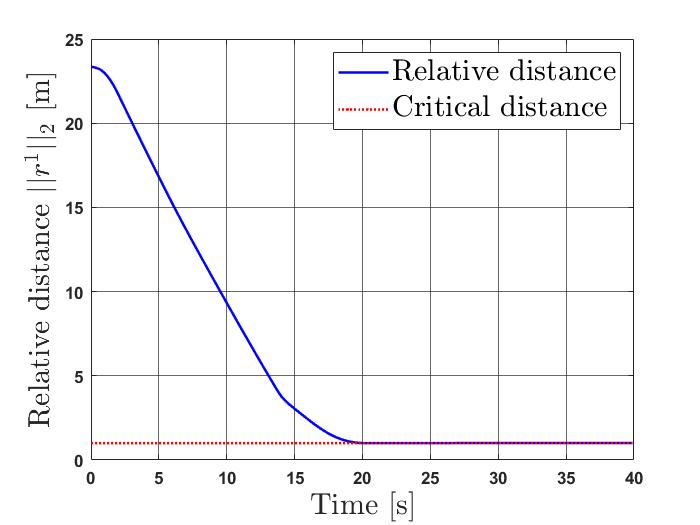}
    \caption{}
    \label{fig:reldist_s1}
\end{subfigure}
\hfill
\begin{subfigure}{0.32\textwidth}
    \centering
    \includegraphics[width=\linewidth]{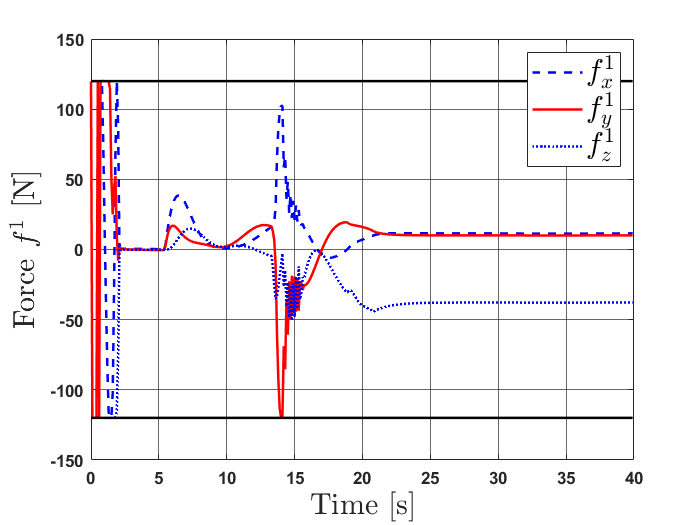}
    \caption{}
    \label{fig:force_s1}
\end{subfigure}
\hfill
\begin{subfigure}{0.32\textwidth}
    \centering
    \includegraphics[width=\linewidth]{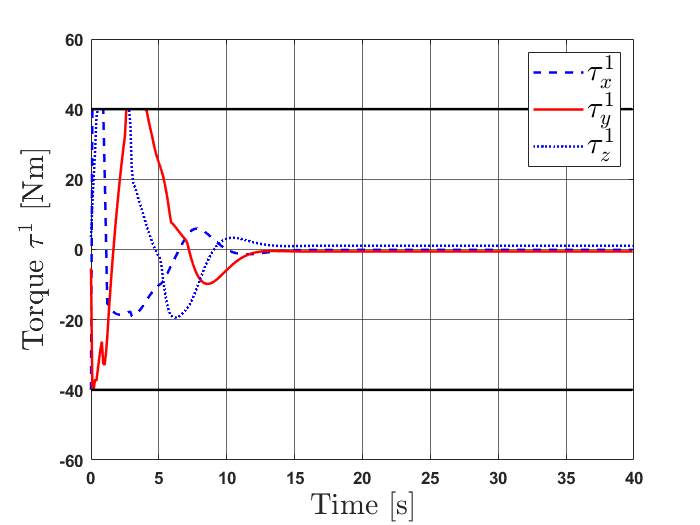}
    \caption{}
    \label{fig:torque_s1}
\end{subfigure}

\caption{Time history of servicer 1's relative distance and inputs throughout the RVD. (a)--(c) shows the relative distance, force, and torque respectively. (a) shows the relative distance respecting the imposed target collision avoidance constraint. (b)--(c) shows the force and torque inputs satisfying the imposed input constraints.}
\label{fig:inputs_reldist}
\end{figure}

\subsection{Results}
As can be seen in Figure~\ref{fig:poserr}, all of the servicers relative position errors converges to zero after about 20-25 seconds into the simulation, and the relative velocities in Figure~\ref{fig:velerr} converges around 25 seconds. In Figures~\ref{fig:atterr}-\ref{fig:avelerr}, the servicers attitude and angular velocities error converges faster to zero due to no external disturbances nor coupling effects which allows for fast tracking throughout the procedure.
\begin{figure}[!hbtp]
\centering
\begin{subfigure}{0.33\columnwidth}
    \centering
    \includegraphics[width=\linewidth]{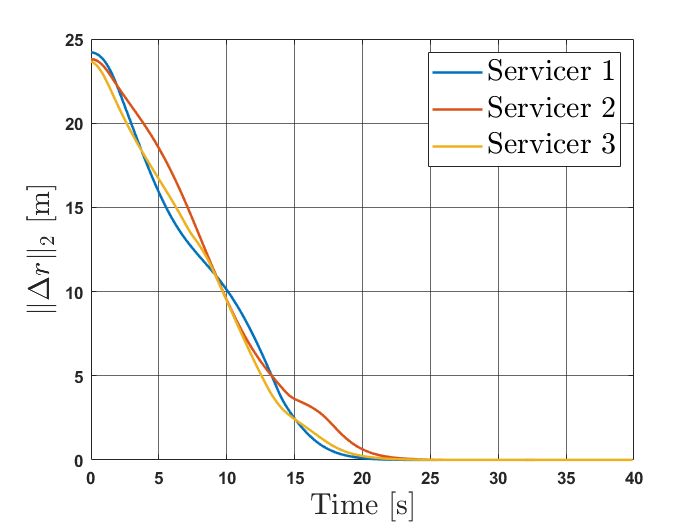}
    \caption{}
    \label{fig:poserr}
\end{subfigure}
\begin{subfigure}{0.33\columnwidth}
    \centering
    \includegraphics[width=\linewidth]{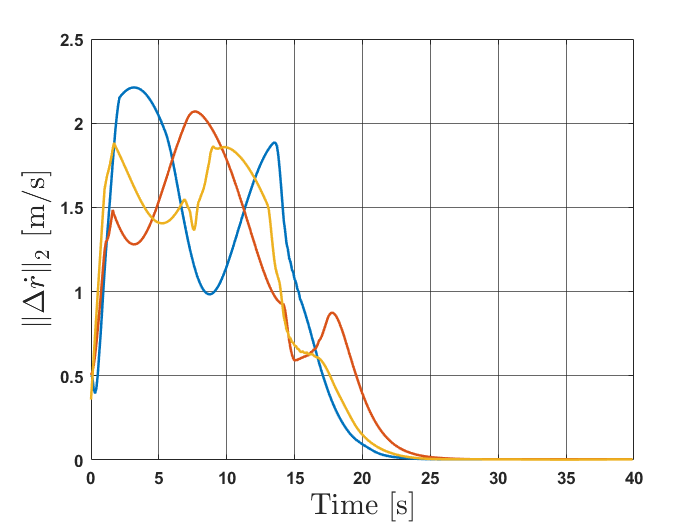}
    \caption{}
    \label{fig:velerr}
\end{subfigure}
\par
\begin{subfigure}{0.33\columnwidth}
    \centering
    \includegraphics[width=\linewidth]{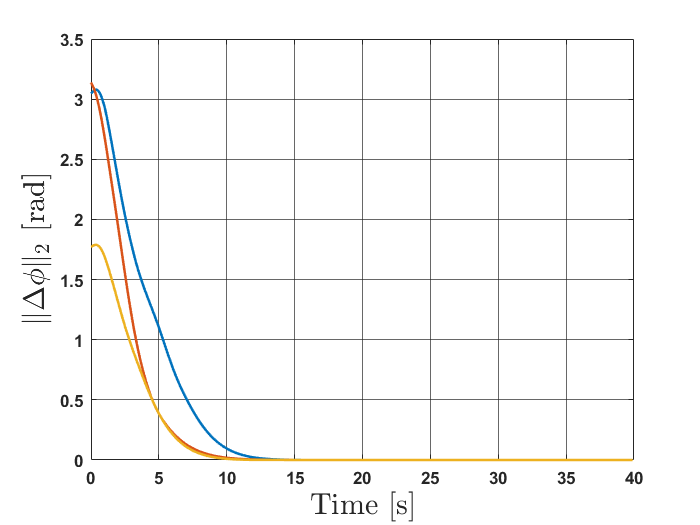}
    \caption{}
    \label{fig:atterr}
\end{subfigure}
\begin{subfigure}{0.33\columnwidth}
    \centering
    \includegraphics[width=\linewidth]{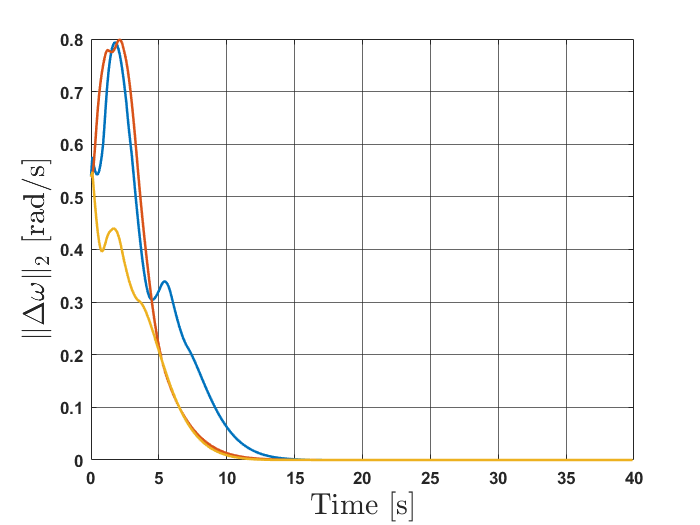}
    \caption{}
    \label{fig:avelerr}
\end{subfigure}
\caption{Time history of the servicers tracking state errors. (a)-(d) shows the position, velocity, attitude, and angular velocity errors converging to zero respectively, between the servicers and respective target docking coordinates.}
\label{fig:servicer_errors}
\end{figure}
\begin{figure}[!hbtp]
\centering
\begin{subfigure}{0.33\columnwidth}
    \centering
    \includegraphics[width=\linewidth]{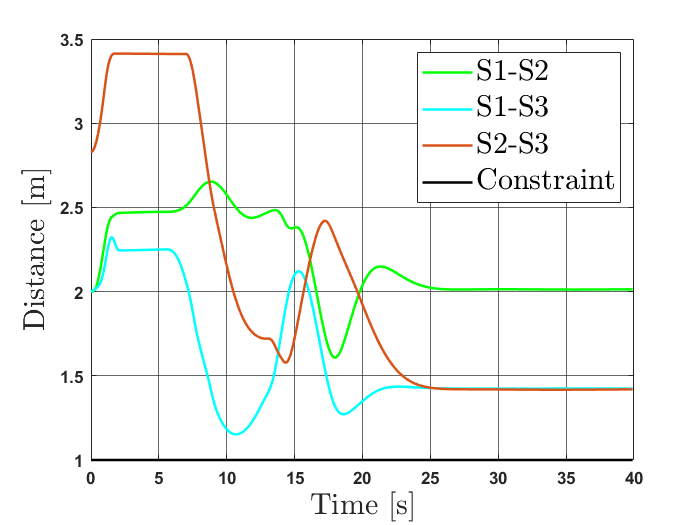}
    \caption{}
    \label{fig:servicer_sep}
\end{subfigure}
\begin{subfigure}{0.33\columnwidth}
    \centering
    \includegraphics[width=\linewidth]{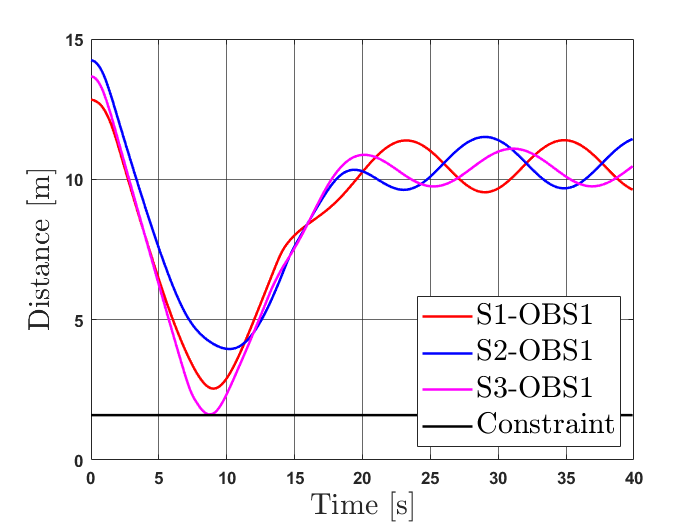}
    \caption{}
    \label{fig:obs1_sep}
\end{subfigure}
\par
\begin{subfigure}{0.33\columnwidth}
    \centering
    \includegraphics[width=\linewidth]{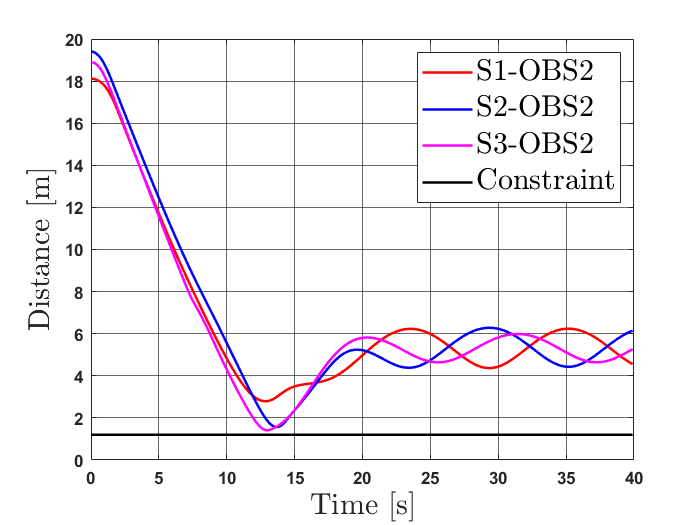}
    \caption{}
    \label{fig:obs2_sep}
\end{subfigure}
\begin{subfigure}{0.33\columnwidth}
    \centering
    \includegraphics[width=\linewidth]{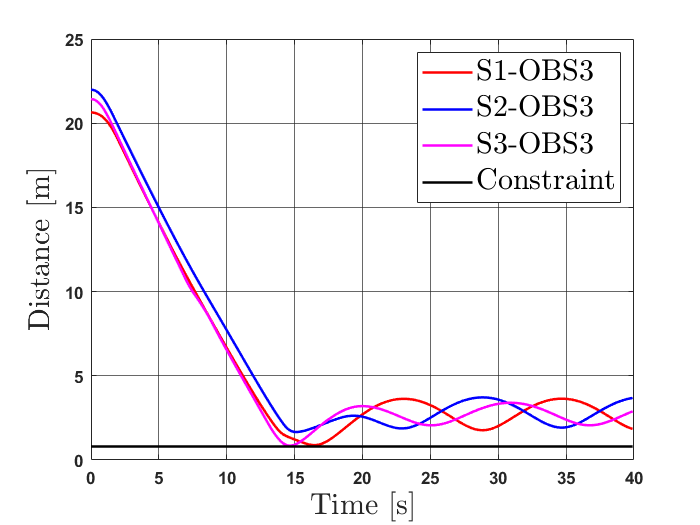}
    \caption{}
    \label{fig:obs3_sep}
\end{subfigure}
\caption{Time history of non-target separation distances. (a) shows the separation distance between servicers during the approach to their respective target docking points. (b)-(d) shows the separation distance between the servicers and obstacle 1-3 during RVD to the target respectively.}
\label{fig:separation_distances}
\end{figure}
As can be seen in the results of the $1^{\mathrm{st}}$ servicer, the position tracking in Figure~\ref{fig:pos_s1} performs well and manages to respect the target collision avoidance constraint as can be seen in Figure~\ref{fig:reldist_s1}. The relative velocity of the servicers are constrained in order to satisfy soft-docking conditions which is also respected and seen in Figure~\ref{fig:vel_s1}. Servicer 1 is initialized to be aligned with the inertial frame and it can be seen in Figure~\ref{fig:att_s1} and Figure~\ref{fig:avel_s1} that the attitude and angular velocity tracks faster than the translational parts, and the angular velocity limit of the servicers is respected as well.

The force and torque inputs in the servicers body frames can be seen in Figures~\ref{fig:force_s1}--\ref{fig:torque_s1} respectively, where they remain inside of the constrained limits. Also here it can be clearly seen that the torque converges faster whereas the force takes longer before stabilizing around non-zero force components needed to maintain tracking of the relative position and velocity of the targets docking port.

In Figure~\ref{fig:separation_distances} the separation distances between servicer-servicer can be seen in Figure~\ref{fig:servicer_sep} and servicer-obstacle distances can be seen in Figures~\ref{fig:obs1_sep}-\ref{fig:obs3_sep}. None of the separation constraints are violated and all of the servicers maintain a safe distance between each other and show the capability to avoid the fixed obstacles throughout the RVD towards the target. For further visualization purposes an animation of the RVD simulation was created, and a full video is provided on \url{https://youtu.be/C3gmMvoiY-o}. 

\section{Conclusions \& future work}
A proof-of-concept limited information sharing decentralized control strategy using a 6-DOF MPC for multiple servicers to RVD with a tumbling target has been presented in this paper. The model manages to satisfy all of the constraints presented using an LTI framework which allows for a simple and fast model. The simulation was successful when tested with different arrangements of the fixed obstacles with the same constraint requirements.

Future work includes off-setting the servicers docking ports from their center-of-masses, and also coupling it to the servicers attitude. Implementation of conical constraints, such as line-of-sight constraints, should be investigated in order to test if it is applicable to vision-based estimation. A more realistic approach to the safety boundaries such as using ellipsoidal envelopes for solar panels should be implemented. Furthermore, to assess the robustness of the LTI formulation, uncertainty in the servicers inertia matrix $J$ should be investigated. Overall robustness studies, such as Monte Carlo simulations using different initial conditions and disturbances should be analyzed for further validation.

% \section*{APPENDIX}

% \section*{ACKNOWLEDGMENT}

\bibliographystyle{IEEEtran}
\bibliography{references}

\end{document}